%% file: root.tex
\documentclass[letterpaper, 10 pt, conference]{ieeeconf}
\IEEEoverridecommandlockouts
\usepackage[utf8]{inputenc} 
\usepackage{dsfont}

\usepackage{amsfonts}
\usepackage{amsmath}
\usepackage{amssymb}
\usepackage{hyperref}       
\usepackage{url}            
\usepackage{booktabs}       

\usepackage{nicefrac}       
\usepackage{microtype}      
\usepackage{algorithmic}
\usepackage{mathrsfs}
\usepackage{graphicx}
\usepackage{textcomp}
\usepackage{xcolor}
\usepackage{subcaption}

\newtheorem{lemma}{Lemma}

\newtheorem{assumption*}{Assumption}
\newtheorem{stdassumption*}{Standing Assumption}

\newtheorem{proposition}{Proposition}
\newtheorem{definition}{Definition}
\newtheorem{definition*}{Definition}

\usepackage{mathtools}
\graphicspath{ {./images/} }

\usepackage[ruled,vlined]{algorithm2e}

\usepackage{cleveref}
\def\BibTeX{{\rm B\kern-.05em{\sc i\kern-.025em b}\kern-.08em
    T\kern-.1667em\lower.7ex\hbox{E}\kern-.125emX}}
    
\begin{document}

\title{\bf \LARGE 
  Actionable Recourse in Competitive Environments:\\ A Dynamic Game of Endogenous Selection
}

\author{Ya-Ting Yang and Quanyan Zhu
\thanks{Authors are with the Department of Electrical and Computer Engineering, New York University, NY, 11201, USA {\tt \{yy4348, qz494\}@nyu.edu}}
}

\maketitle

\input{abstract}

\input{introduction}

\input{relatedwork}

\input{risk}

\input{actionable}

\input{dynamic}

\input{experiment}

\input{conclusion}

\bibliographystyle{abbrv}
\bibliography{ref}

\end{document}

%% file: abstract.tex
\begin{abstract}
Actionable recourse studies whether individuals can modify feasible features to overturn unfavorable outcomes produced by AI-assisted decision-support systems. However, many such systems operate in competitive settings, such as admission or hiring, where only a fraction of candidates can succeed. A fundamental question arises: what happens when actionable recourse is available to everyone in a competitive environment?
This study proposes a framework that models recourse as a strategic interaction among candidates under a risk-based selection rule. Rejected individuals exert effort to improve actionable features along directions implied by the decision rule, while the success benchmark evolves endogenously as many candidates adjust simultaneously. This creates endogenous selection, in which both the decision rule and the selection threshold are determined by the population's current feature state. This interaction generates a closed-loop dynamical system linking candidate selection and strategic recourse. We show that the initially selected candidates determine both the benchmark of success and the direction of improvement, thereby amplifying initial disparities and producing persistent performance gaps across the population.
\end{abstract}

%% file: introduction.tex
\section{Introduction}
\label{sec:intro}

Artificial intelligence (AI) systems are increasingly embedded in institutional decision-making processes. These systems generate predictions, risk scores, and recommendations that influence access to resources, opportunities, and rights. In such settings, AI functions as a decision-support mechanism that shapes institutional outcomes.
Within this context, actionable recourse concerns whether individuals can meaningfully respond to unfavorable decisions produced by AI systems \cite{ustun2019actionable,karimi2020model,karimi2021algorithmic}. Considering a decision support system that maps an individual’s profile to classification results, actionable recourse exists if an individual can make changes to their profile that would lead the system to produce a favorable outcome.
It complements existing principles of AI governance \cite{batool2025ai}, such as explainability, fairness, and accountability, by providing the perspective that a responsible decision system must ensure that outcomes depend, at least in part, on mutable features that individuals can modify with feasible effort. 

However, AI decision-support systems often operate in scenarios involving many individuals simultaneously. In domains such as university admissions, job search, and competitive grants, environments are inherently competitive, meaning that the improvements of some individuals can affect the outcomes of others.
This raises a fundamental question: what happens when actionable recourse is available to everyone in a competitive environment? If all individuals follow the recommended improvement paths, the collective response may alter the distribution of features and shift the decision boundary. This refers to ``endogenous selection'', a mechanism in which both the decision rule and the selection threshold are determined by the population's state. In such settings, improvement becomes relative rather than absolute.

Hence, in this work, we model actionable recourse as a strategic interaction among candidates. The system designer employs a risk-based selection rule. Individuals who receive unfavorable outcomes exert effort with cost to improve actionable features. In a dynamic setting, when many candidates adjust simultaneously, the success benchmark changes. The system evolves toward an equilibrium where further improvement becomes prohibitively costly or ineffective.
This dynamic is closely related to the sociological concept of ``involution'' in competitive systems: as individuals invest increasing effort to improve their profiles, the collective benchmark rises. We also show that individuals who initially receive favorable outcomes define not only the standard of success but also the direction in which others must compete. In this case, competitive recourse can amplify disparities and generate persistent performance gaps across the population.

\noindent\textbf{Contributions:} This work proposes a formal framework for analyzing actionable recourse as a dynamic and strategic process in a competitive environment. The model captures repeated interactions between a designer who selects decision rules and candidates who respond by improving actionable features. We analyze the resulting dynamical system, establish the equilibrium states in which candidates cannot profitably improve through feasible effort, and characterize how competition and resource constraints shape long-run outcomes. 

%% file: relatedwork.tex
\section{Literature Review}

Our work builds on the literature on counterfactual explanations and actionable recourse in classification, which studies whether individuals can modify controllable features to reverse unfavorable decisions \cite{verma2024counterfactual}. While counterfactual explanations identify hypothetical feature changes that would alter a model’s prediction, actionable recourse focuses on whether individuals can realistically implement such changes through feasible and controllable actions.

Ustun et al.~\cite{ustun2019actionable} formalize actionable recourse in linear classification and propose algorithms for computing minimal-cost feature changes under feasibility constraints. Subsequent work extends recourse to model-agnostic and causal settings \cite{karimi2020model,karimi2021algorithmic}, and studies issues such as feasibility, robustness, and the cost of recourse \cite{dominguez2022adversarial,pawelczyk2020counterfactual}. 
The proposed framework is also closely related to strategic classification, which models how agents modify their features in response to a classifier \cite{hardt2016strategic,milli2019social}. These works highlight welfare trade-offs and social costs arising from strategic adaptation. Closely related is the literature on performative prediction, which studies how the deployment of predictive models alters the underlying data-generating process, creating feedback between model predictions and population behavior \cite{perdomo2020performative,drusvyatskiy2023stochastic}. Our work complements these perspectives by integrating recourse with strategic and performative feedback in a competitive and resource-constrained environment.

Finally, the proposed framework connects to the broader literature on competition for limited resources. Resource allocation has been extensively studied in communication networks \cite{kelly1998rate,zhang2011dynamic}, IoT and robotics systems \cite{zhao2024incentive,farooq2021resource,farooq2020qoe}, contest theory \cite{tullock1980efficient}, bargaining \cite{nash1950bargaining}, and market design. We bring this perspective into AI-mediated decision systems and show that competitive recourse dynamics can produce equilibrium stratification and limits on upward mobility when institutional resources are constrained.

%% file: risk.tex
\section{Risk-based Competitive Selection}

In many competitive allocation problems, a decision-maker ranks a population and selects only a fixed fraction, $\rho \in (0,1)$, of top candidates, as in admissions, hiring, and promotion. In these settings, ground-truth labels are typically unavailable; instead, the designer defines a scoring rule whose upper tail determines the selected group.

Consider a population of $n$ candidates represented by feature vectors $X=\{x_i\}_{i=1}^n \subset \mathbb{R}^d$ and assume that $\rho n \in \mathbb{N}$. A linear scoring rule can be defined by assigning each candidate a score $s_i = w^\top x_i$, where $w \in \mathbb{R}^d$. Let $P_n$ denote the empirical distribution assigning mass $1/n$ to each $x_i$, and let $s_{(1)} \ge \dots \ge s_{(n)}$ denote the order statistics of $\{s_i\}_{i=1}^n$. Then, selecting the top $\rho$ fraction corresponds to maximizing the average score of the $\rho n$ values. That is, $\frac{1}{\rho n}\sum_{i=1}^{\rho n} s_{(i)}$.

\subsection{Upper-Tail CVaR Representation}

Conditional Value-at-Risk (CVaR) is a risk measure that evaluates the expected value within the tail of a distribution. While the standard definition $\mathrm{CVaR}_\alpha$  measures the expected loss in the worst $(1-\alpha)\%$ cases, here we use a customized $\mathrm{CVaR}_\rho^{\mathrm{up}}$ to characterize the average score among the highest-performing $\rho$ fraction of the population. Here, we define
\begin{equation}
\label{eq:upper_cvar_formal}
\mathrm{CVaR}_\rho^{\mathrm{up}}(w)=\min_{\eta \in \mathbb{R}}\Big\{\eta+\frac{1}{\rho n}\sum_{i=1}^n(w^\top x_i - \eta)_+\Big\}.
\end{equation} 

\begin{lemma}[CVaR representation]
For every $w \in \mathbb{R}^d$, we have
$\mathrm{CVaR}_\rho^{\mathrm{up}}(w)=\frac{1}{\rho n}\sum_{i=1}^{\rho n} s_{(i)}$.
\label{lemma:CvaR}
\end{lemma}

\begin{proof}
Fix $w$ and define $\phi(\eta)=\eta+\frac{1}{\rho n}\sum_{i=1}^n (s_{i}-\eta)_+$, where $s_{i} = w^\top x_{i}$. The function $\phi$ is convex and piecewise linear in $\eta$. Assuming distinct scores for simplicity, then $\partial \phi(\eta)=1-\frac{1}{\rho n}\#\{ i : s_i \ge \eta \}$. An optimizer $\eta^\star$ must satisfy $0 \in \partial \phi(\eta^\star)$, which implies $\#\{ i : s_{i} \ge \eta^\star \}=\rho n$.
Hence, $n-\rho n$ scores lie below $\eta^\star$, so $\eta^\star$ coincides with the $(\rho n)$-th largest score, i.e., $\eta^\star = s_{(\rho n)}$.
Substituting $\eta^\star$ into $\phi$ yields $\phi(\eta^\star)=\frac{1}{\rho n}\sum_{i=1}^{\rho n}s_{(i)}$, which proves the claim.
\end{proof}

This representation shows that competitive selection can be written as the maximization of the customized CVaR of the score distribution. Intuitively, the objective focuses on the expected score among the top-performing $\rho$ fraction of the population, providing a convenient optimization-based representation of top-$\rho$ selection.

\subsection{Regularized Optimization Problem} \label{sec:regularized optimization}

\begin{lemma}
Let $s_i(w)=w^\top x_i$ and let $s_{(1)}(w)\ge \dots \ge s_{(n)}(w)$ denote the order statistics. 
Then, for $t\ge 0$, we have $\mathrm{CVaR}_\rho^{\mathrm{up}}(t w)=t\, \mathrm{CVaR}_\rho^{\mathrm{up}}(w).$
\end{lemma}

\begin{proof}
Since $s_i(tw)=(tw)^\top x_i=t s_i(w)$, scaling $w$ by $t\ge0$ preserves the score ordering. Therefore,
$\mathrm{CVaR}_\rho^{\mathrm{up}}(tw)=\frac{1}{\rho n}\sum_{i=1}^{\rho n} s_{(i)}(tw)=\frac{1}{\rho n}\sum_{i=1}^{\rho n} t s_{(i)}(w)=t\,\mathrm{CVaR}_\rho^{\mathrm{up}}(w)$.
\end{proof}

This shows that $\mathrm{CVaR}_\rho^{\mathrm{up}}$ is positively homogeneous of degree one. If there exists some $w$ such that $\mathrm{CVaR}_\rho^{\mathrm{up}}(w)>0$, then $\sup_{w}\mathrm{CVaR}_\rho^{\mathrm{up}}(w)=+\infty$, since $\mathrm{CVaR}_\rho^{\mathrm{up}}(tw)=t \mathrm{CVaR}_\rho^{\mathrm{up}}(w)$ for all $t\ge 0$. Thus, to make the maximization problem well posed, we need to impose a normalization or regularization on $w$.
Therefore, we consider quadratic regularization with parameter $\lambda>0$:
\begin{equation}
\label{eq:regularized_problem}
\max_{w \in \mathbb{R}^d} \mathrm{CVaR}_\rho^{\mathrm{up}}(w)-\frac{\lambda}{2}\|w\|^2.
\end{equation}
Using \eqref{eq:upper_cvar_formal}, this is equivalent to
\begin{equation}
\label{eq:expanded_problem}
\max_{w}\bigg[\min_{\eta}\big(\eta+\frac{1}{\rho n}\sum_{i=1}^n(w^\top x_i-\eta)_+\big)-\frac{\lambda}{2}\|w\|^2\bigg].
\end{equation}

\subsubsection{Dual Representation of the Inner Problem}

To analyze problem \eqref{eq:expanded_problem}, we introduce slack variables $u_i \ge 0$ satisfying $u_i \ge w^\top x_i-\eta$. Then, for fixed $w$, the inner problem becomes
\begin{equation}
\label{eq:inner_lp}
R(w)=\min_{\eta,u} \eta+\frac{1}{\rho n}\sum_{i=1}^n u_i, \text{s.t. } u_i \ge w^\top x_i-\eta,\; u_i \ge 0.
\end{equation}
As Slater's condition holds for \eqref{eq:inner_lp}, strong duality applies. We then let $\alpha_i \ge 0, \beta_i \ge 0$ be the dual variables associated with the constraints $w^\top x_i-\eta-u_i \le 0$ and $-u_i \le 0$, respectively. Forming the Lagrangian and minimizing over $(\eta,u)$ yields us the dual constraints: $\sum_{i=1}^n \alpha_i = 1,\ 0 \le \alpha_i \le \frac{1}{\rho n}, \forall i$.
Therefore, $R(w)$ can be written in dual as
\begin{equation}
\label{eq:inner_dual}
\max_{\alpha}\Big\{\sum_{i=1}^n \alpha_i w^\top x_i: 0 \le \alpha_i \le \frac{1}{\rho n},\sum_{i=1}^n \alpha_i = 1 \Big\}.
\end{equation}

\subsubsection{Equivalent Max--Max Formulation}

By combining \eqref{eq:regularized_problem} and \eqref{eq:inner_dual}, we can obtain
\begin{equation}
\begin{aligned}
\label{eq:max_max_problem}
\max_{w,\alpha}\ & w^\top \Bigl(\sum_{i=1}^n \alpha_i x_i\Bigr)-\frac{\lambda}{2}\|w\|^2\ \\
\text{s.t. } &0 \le \alpha_i \le \frac{1}{\rho n},\ \sum_{i=1}^n \alpha_i = 1.
\end{aligned}
\end{equation}
For any fixed feasible $\alpha$ that satisfies the constraints in \eqref{eq:max_max_problem}, the objective is strictly concave in $w$, and its unique maximizer is $w^\star(\alpha)=\frac{1}{\lambda}\sum_{i=1}^n \alpha_i x_i$.
Substituting $w^\star(\alpha)$ back yields the reduced problem
\begin{equation}
\label{eq:alpha_reduced_problem}\max_{\alpha}\ \frac{1}{2\lambda}\bigg\|\sum_{i=1}^n \alpha_i x_i \bigg\|^2 \ \text{s.t. } 0 \le \alpha_i \le \frac{1}{\rho n},\ \sum_{i=1}^n \alpha_i = 1.
\end{equation}

\subsubsection{Structure of the Optimal Solution}

By complementary slackness, if $\alpha_i^\star>0$, then $u_i^\star = w^{\star\top}x_i-\eta^\star \ge 0$, so $w^{\star\top}x_i \ge \eta^\star$.
Therefore, we can say that points with positive dual weight lie in the upper tail determined by the threshold $\eta^\star$.
The feasible set for $\alpha$ is a capped simplex. Since each coordinate is bounded above by $1/(\rho n)$, at least $\rho n$ coordinates must be positive to satisfy the constraint. Extreme points therefore have exactly $\rho n$ coordinates equal to $1/(\rho n)$ and the rest zero. Hence, for a nondegenerate dataset, $\alpha_i^\star \in \{0,1/(\rho n)\}$, and $w^\star=\frac{1}{\lambda\rho n}\sum_{i \in \mathcal I} x_i$,
where $\mathcal I$ denotes the subset of cardinality $\rho n$ corresponding to the selected points under the optimal score function.
Thus, competitive selection can be interpreted as choosing a subset of mass $\rho n$ whose empirical centroid has the largest norm after quadratic regularization. The optimal scoring vector is proportional to the centroid of this selected upper-tail subset, while the dual variables act as endogenous selection weights on a capped simplex.

%% file: actionable.tex
\section{Actionable Recourse Under Risk-Based Selection}

The competitive selection rule derived above determines an optimal scoring direction $w^\star$ and threshold $\eta^\star$. A candidate with $x_i$ is selected if and only if $w^{\star\top} x_i \ge \eta^\star$.
Because the rule is linear, explicit, and geometrically interpretable, it induces a well-defined separating hyperplane. This structure naturally gives rise to the notion of \emph{actionable recourse}: for any rejected candidate satisfying $w^{\star\top} x_i < \eta^\star$, one may ask what minimal feasible modification of $x_i$ would move the candidate across the competitive boundary.

\subsection{Actionable and Immutable Features}

We partition the feature indices $\{1,\dots,d\} = J_A \cup J_N$, where $J_A$ denotes \emph{actionable} coordinates and $J_N$ denotes \emph{immutable} coordinates. An admissible action $a_i \in \mathbb{R}^d$ must lie in the subspace $\mathcal{A} := \{a \in \mathbb{R}^d : a_j = 0 \text{ for } j \in J_N\}$. In addition, we also let $w_A^\star$ be the projection of the optimal weights onto $J_A$ for subsequent analysis.

\subsection{Minimal Recourse and Competitive Margin}\label{sec:minimal_recourse}

Suppose a candidate with $x_i$ is not selected ($w^{\star\top} x_i < \eta^\star$), we define the \emph{competitive margin} as $\Delta_i := \eta^\star - w^{\star\top} x_i$, which is greater than zero. If the candidate applies an action $a_i \in \mathcal A$, their new score becomes $w^{\star\top}(x_i + a_i)=w^{\star\top}x_i + w_A^{\star\top} a_i$. Consequently, the condition for successful recourse is: $w_A^{\star\top} a_i \ge \Delta_i$. 
Then, to measure the burden of change, we assume a quadratic cost function, reflecting the intuition that effort becomes disproportionately more difficult as the scale of change increases. The minimal recourse problem is:
\begin{equation}
\label{eq:recourse_problem}
\min_{a_i \in \mathcal{A}} \frac{1}{2}\|a_i\|^2 \quad \text{s.t.} \quad w_A^{\star\top} a_i \ge \Delta_i.
\end{equation}

\begin{proposition}[Closed-Form Recourse]
If $w_A^\star \neq 0$, the unique optimal action and cost for \eqref{eq:recourse_problem} are:
\[
a_i^\star = \frac{\Delta_i}{\|w_A^\star\|^2} w_A^\star, \qquad c_i^\star = \frac{\Delta_i^2}{2\|w_A^\star\|^2}.
\]
\end{proposition}
\begin{proof}
    The Lagrangian is $\mathcal{L}(a_i,\lambda)=\frac12\|a_i\|^2+\lambda(\Delta_i - w_A^{\star\top} a_i)$ with multiplier $\lambda \ge 0$. In order to find the minimum, we take the derivative with respect to the action $a_i$ and set it to zero, leading to $a_i = \lambda w_A^\star$. Then, since $\Delta_i>0$, the inequality constraint must bind at optimality, leading to $w_A^{\star\top} a_i = \lambda \|w_A^\star\|^2 = \Delta_i$, $\lambda=\Delta_i/\|w_A^\star\|^2$. Substituting $\lambda$ back into the expression for $a_i$ concludes the proof.
\end{proof}
Note that The cost $c_i^\star = \frac{\Delta_i^2}{2\|w_A^\star\|^2}$ grows quadratically in the margin, which indicates that candidates farther from the threshold face disproportionately larger effort. Moreover, the denominator $\|w_A^\star\|^2$ can be interpret as the measure of the \emph{actionability strength}. If actionable weight is rather small, recourse becomes expensive even for small margins.

\subsection{Structural Accountability}

This framework reveals the following \emph{structural accountability failure}:
\begin{itemize}
    \item \textbf{Infeasibility:} If $w_A^\star = 0$ and $\Delta_i > 0$, the decision boundary becomes insensitive to any admissible individual effort and recourse becomes impossible. The rejection result is final in a structural sense rather than a performance-based sense. This denies contestability.
    \item \textbf{Actionability Strength:} The term $\|w_A^\star\|^2$ diagnosticates the system's permeability. Low actionable weight signals a ``practically immutable'' system that performs rigid sorting rather than evaluative selection.
    \item \textbf{Incentive Alignment:} Since $a_i^\star \propto w_A^\star$, the classifier explicitly directs and therefore incentivizes effort. Legitimacy may require that $w_A^\star$ aligns with socially or economically productive dimensions.
\end{itemize}

%% file: dynamic.tex
\section{Dynamic Recourse with Endogenous Direction}


We consider discrete time and indexed by $t \in \mathbb{N} := \{0,1,2,\dots\}$. At each time $t$, the state of the system is the population feature profile $X^t = \{x_i^t\}_{i=1}^n \subset \mathbb{R}^d$, with associated empirical distribution $P_n^t$ supported by $X^t$. Thus, the population configuration constitutes the state variable. The evolution of the system is driven by the interaction between the designer's selection rule and the candidates' responses.

\subsection{Designer’s CVaR-Based Selection Problem}\label{sec:designer}

At time $t$, the designer observes $X^t$ and solves the risk-based selection problem described in \ref{sec:regularized optimization}. This produces an optimal scoring direction $w^{\star,t}$ and threshold $\eta^{\star,t}$ determined endogenously by the current feature distribution. The selection results are then announced: candidates satisfying $w^{\star,t\top} x_i^t \ge \eta^{\star,t}$ are accepted, while others are rejected. Let 
For each rejected candidate, the system computes the minimal actionable recourse direction derived from $w^{\star,t}$. Since the scoring rule is linear, the improvement direction coincides with the actionable projection $d^t := \Pi_{J_A}(w^{\star,t})$ of the currently most competitive profile. This direction is endogenous: it is not externally specified, but is induced by the geometry of the CVaR-based selection. Each rejected candidate is then informed of the direction in the feature space along which the improvement efficiently increases their score. 

\subsection{Candidates' Optimization Problem: Minimal Recourse}
 
The candidates then choose the effort along the direction $d^t$. In the simplest deterministic formulation, each rejected candidate updates their feature vector according to the minimal recourse prescription in \ref{sec:minimal_recourse} (or another rule depending only on $d^t$ and their current margin), while the accepted candidates can remain unchanged or follow an exogenous adjustment rule. The resulting adjustments produce the next feature profile $X^{t+1} = \{x_i^{t+1}\}_{i=1}^n$.
This interaction defines a deterministic closed-loop dynamical system $X^{t+1} = \Phi(X^t)$,
where the mapping $\Phi$ captures both the designer’s optimization step and the candidates’ best-response recourse step. The classifier and the population co-evolve: the selection rule depends on the current feature distribution, while the distribution evolves in response to that rule. 

\subsection{Candidates' Problem with Logarithmic Barrier Effort Cost}\label{sec:log_barrier}

Consider candidate $i$ with feature vector $x_i^t \in \mathcal X \subset \mathbb R^d$, where $\mathcal X$ is nonempty, compact, and convex. Let $e_A$ Assume there exists a distinguished actionable coordinate $g_i^t := e_A^\top x_i^t$ that candidate $i$ is allowed to modify. This coordinate is subject to a physical or institutional ceiling, such that $g_i^t \le \overline g$ (e.g., a maximum score of $100$ on an examination). We define the ``improvement gap'' as: as $\Delta_i^t := \overline g - g_i^t > 0$. The quantity $\Delta_i^t$ measures the distance of the actionable feature from its maximal attainable value, representing the remaining capacity for upward mobility in that dimension.

\subsubsection{Admissible Actions}
At time $t$, the system communicates an actionable direction $d^t \in \mathbb R^d$. Since recourse operates through modifiable coordinates, we can assume $e_A^\top d^t > 0$, then the normalized actionable direction is defined as $\Tilde{d}^t=d^t/e_A^\top d^t$. Admissible actions are restricted to movements along direction $\Tilde{d}^t$, so that $a_i^t = \gamma_i^t \Tilde{d}^t$ with $\gamma_i^t \ge 0$ for candidate $i$. Moreover, feasibility requires that the improved feature vector remains within the domain $x_i^t + \gamma_i^t \Tilde{d}^t \in \mathcal X$, which implies the effort is bounded by the remaining improvement gap: $\gamma_i^t < \Delta_i^t$. The compactness of $\mathcal X$ ensures that the admissible set for $\gamma_i^t$ is bounded.

\subsubsection{Logarithmic Barrier Effort Cost}
To model increasing difficulty near the ceiling $\overline g$, we encode the boundary constraint through a logarithmic barrier. Let $k_i>0$ and $\theta_i>0$ be fixed parameters. The effort cost is defined as 
\begin{equation}
    C_i^t(\gamma)=\frac{k_i}{2}\gamma^2-\theta_i \log(\Delta_i^t - \gamma),\quad 0 \le \gamma < \Delta_i^t.
\label{eq:effort_cost}
\end{equation}
The quadratic term captures standard convex effort cost, while the logarithmic barrier enforces the domain constraint smoothly. The function $C_i^t$ is strictly convex on $(0,\Delta_i^t)$ since $\partial^2 C_i^t/\partial \gamma^2 = k_i + [\theta_i/(\Delta_i^t - \gamma)^2] > 0$. Moreover, as $\gamma \to \Delta_i^t$, the marginal cost $\partial C_i^t/\partial \gamma = k_i \gamma + [\theta_i/(\Delta_i^t - \gamma)]$ diverges to $+\infty$. Thus improvement becomes infinitely costly as the actionable feature approaches its ceiling.

\subsubsection{Candidate Optimization Problem}\label{sec:effort_optimization}
For a rejected candidate $i \in \mathcal R^t$, the objective is to maximize perceived benefit from score improvement net of effort cost: $$\max_{0 \le \gamma_i^t < \Delta_i^t}\;\beta_i w^{t\top}(x_i^t + \gamma_i^t \Tilde{d}^t)-\frac{k_i}{2}(\gamma_i^t)^2+\theta_i \log(\Delta_i^t - \gamma_i^t),$$ where $\beta_i>0$ is the valuation of score improvement. Using linearity, $w^{t\top}(x_i^t + \gamma_i^t \Tilde{d}^t) = w^{t\top} x_i^t + \gamma_i^t w^{t\top} \Tilde{d}^t$. Discarding constants yields the reduced problem $$\max_{0 \le \gamma_i^t < \Delta_i^t}\;\beta_i \gamma_i^t w^{t\top} \Tilde{d}^t-\frac{k_i}{2}(\gamma_i^t)^2+\theta_i \log(\Delta_i^t - \gamma_i^t).$$ The objective is strictly concave in $\gamma_i^t$; therefore, there exists a unique maximizer.

Let $S_i^t := \beta_i w^{t\top} \Tilde{d}^t$. The first-order condition for an interior solution is $S_i^t=k_i \gamma_i^t+[\theta_i/(\Delta_i^t - \gamma_i^t)]$.
Multiplying both sides by $(\Delta_i^t - \gamma_i^t)$ yields the quadratic equation $k_i (\gamma_i^t)^2-(k_i \Delta_i^t + S_i^t)\gamma_i^t-(\theta_i - S_i^t \Delta_i^t)=0$.
The admissible root satisfying $0 \le \gamma_i^t < \Delta_i^t$ is
$$\gamma_i^{t\star}=\frac{(k_i \Delta_i^t + S_i^t)-\sqrt{(k_i \Delta_i^t + S_i^t)^2 + 4k_i(\theta_i - S_i^t \Delta_i^t)}
}{2k_i}.$$
Recall that $\mathcal I^t$ represents the set of $\rho n$ selected candidates, and let $\mathcal R^t$ denote the set of rejected candidates. The update rule for $i \in \mathcal R^t$ is therefore $x_i^{t+1} = x_i^t + \gamma_i^{t\star} \Tilde{d}^t$. As $g_i^t \to \overline g$, we have $\Delta_i^t \to 0$, and the barrier term forces $\gamma_i^{t\star} \to 0$. Improvement thus slows smoothly and endogenously near the boundary. 

Combining this update with $w^\star$  in Section \ref{sec:regularized optimization} indexed by $t$ yields the closed-loop recursion
\begin{equation}
\label{eq:closed_loop_expanded}
x_i^{t+1} = x_i^t +
\mathbf{1}_{\{ i \in \mathcal R^t \}}
\gamma_i^{t\star}\,\Pi_{J_A}\big(\frac{1}{\lambda \rho n}\sum_{j \in \mathcal I^t} x_j^t\big).
\end{equation}
Equation \eqref{eq:closed_loop_expanded} shows that the evolution of every candidate's feature vector depends on two endogenous objects: the centroid of the current top-$\rho$ tail and the optimally chosen effort $\gamma_i^{t\star}$, which depends on the actionable signal $d^t$ and the remaining gap $\Delta_i^t$. The dynamics are therefore driven by the actionable projection of the tail centroid together with heterogeneous effort responses.

The induced mapping $\Phi : X^t \mapsto X^{t+1}$ is deterministic but piecewise smooth. Nonsmoothness arises from the dependence of $\mathcal I^t$ on the ranking induced by $w^t$. Small perturbations of $X^t$ that change the $\rho n$-th ranked candidate can alter $\mathcal I^t$ discontinuously, changing the centroid and direction $d^t$.
The evolution follows $X^t \to \mathcal I^t \to w^t \to d^t \to \gamma^{t\star} \to X^{t+1}$.


\subsection{Recourse Equilibrium}\label{sec:recourse-equilibrium}

We define the closed-loop mapping $\Phi$ induced by one round of the designer's optimization, optimal recourse, and characterize its fixed points. To avoid notational ambiguity, we explicitly introduce solution operators; the $w,\eta,\gamma$ denote optimizer \emph{points} determined by $X$ through these operators.

We define the
\emph{designer solution operator} $\mathcal D:(\mathbb R^d)^n\to \mathbb R^d\times\mathbb R, \mathcal D(X) := (w_X,\eta_X)$, where $(w_X,\eta_X)$ is the optimizer of \eqref{eq:expanded_problem}.
Given $(w_X,\eta_X)$, the selected and rejected sets are defined as $\mathcal I_X := \{i: w_X^\top x_i \ge \eta_X\}, \mathcal R_X := \{1,\dots,n\}\setminus \mathcal I_X$,
The communicated actionable direction is $d_X := \Pi_{J_A}(w_X)$.
Then, for each $i\in\mathcal R_X$, define the remaining actionable gap
$\Delta_i(X):=\overline g - e_A^\top x_i$ and consider the strictly concave problem in Section \ref{sec:effort_optimization}.
Strict concavity implies a unique maximizer; denote it by $\gamma_{i,X}$.
For $i\in\mathcal I_X$ we set $\gamma_{i,X}:=0$. Then, the \emph{effort response operator} can be defined as $\Gamma_i:(\mathbb R^d)^n\to\mathbb R_+, \Gamma_i(X):=\gamma_{i,X}$.
Finally, define $\Phi:(\mathbb R^d)^n\to(\mathbb R^d)^n$ componentwise by
\begin{equation}
\label{eq:Phi_update_op}
(\Phi(X))_i = x_i
+
\mathbf 1_{\{i\in\mathcal R_X\}} \gamma_{i,X}\, \Tilde d_X,
\quad i=1,\dots,n,
\end{equation}
which yields the closed-loop recursion $X^{t+1}=\Phi(X^t)$.

\begin{definition}[Recourse equilibrium]
A state $X^\star\in(\mathbb R^d)^n$ is a \emph{recourse equilibrium} if
$\Phi(X^\star)=X^\star$.
\label{def:eqm}
\end{definition}

For a recourse equilibrium $X^\star$, write
$(w^\star,\eta^\star):=\mathcal D(X^\star)$, $d^\star:=\Pi_{J_A}(w^\star)$,
$\mathcal I^\star:=\mathcal I_{X^\star}$, and $\mathcal R^\star:=\mathcal R_{X^\star}$.

\begin{proposition}[Fixed-point characterization]
\label{prop:recourse_eq_characterization}
A state $X^\star$ is a recourse equilibrium if and only if $\forall i\in\mathcal R^\star, \gamma_{i,X^\star}\, d^\star = 0$.
\end{proposition}

\begin{proof}
$\Phi(X^\star)=X^\star$ if and only if $(\Phi(X^\star))_i=x_i^\star$ for each $i\in\{1,\dots,n\}$.
Using \eqref{eq:Phi_update_op} with $X=X^\star$ gives $(\Phi(X^\star))_i = x_i^\star + \mathbf 1_{\{i\in\mathcal R^\star\}}\gamma_{i,X^\star}\, \Tilde d^\star$.
Thus. $(\Phi(X^\star))_i=x_i^\star$ holds for all $i$ if and only if $\mathbf 1_{\{i\in\mathcal R^\star\}}\gamma_{i,X^\star}\, \Tilde d^\star=0$ for all $i$, which implies  $\gamma_{i,X^\star}\, d^\star=0$. for $i\in\mathcal R^\star$.
\end{proof}

The fixed-point condition admits two regimes. A \emph{structural equilibrium} occurs when $d^\star=0$, i.e., the learned scoring direction lies entirely in the immutable subspace. An \emph{effort-suppressed equilibrium} occurs when $d^\star\neq 0$ but $\gamma_{i,X^\star}=0$ for all $i\in\mathcal R^\star$, meaning an actionable direction exists but all rejected candidates optimally exert zero effort due to an unfavorable marginal tradeoff (e.g., a strong barrier effect).

\subsection{Emergence of Social Stratification}


\begin{definition}[Social Stratification]
A recourse equilibrium $X^\star$ exhibits \emph{social stratification} if $d^\star=0$ and $\mathcal R^\star\neq\varnothing$.
\end{definition}
Stratification occurs when selection depends only on
immutable coordinates, while some candidates remain rejected.



Let $\bar x_I^t=\frac{1}{|\mathcal I^t|}\sum_{i\in\mathcal I^t}x_i^t$ and $\bar x_R^t=\frac{1}{|\mathcal R^t|}\sum_{i\in\mathcal R^t}x_i^t$ denote the group centroids, and define the intergroup gap $D^t:=\|\bar x_I^t-\bar x_R^t\|$. Although individual members of $\mathcal I^t$ may not move, the centroid $\bar x_I^t$ may evolve over time, i.e., changes in membership of $\mathcal I^t$. Rejected candidates move according to $x_i^{t+1}=x_i^t+\gamma_i^{t\star} \Tilde d^t$ for $i\in\mathcal R^t$.
Then, stratification widens whenever the effective motion of the tail centroid exceeds the improvement of the rejected group, i.e., $\|\bar x_I^{t+1}-\bar x_I^t\| > \|\bar x_R^{t+1}-\bar x_R^t\|$, which implies $D^{t+1}>D^t$. Since both $w^t$ and $d^t$ are determined by $\mathcal I^t$, the top group defines not only the standard of success but also the direction in which others must compete, creating an asymmetric feedback mechanism.

At a stratified equilibrium $X^\star$ with $d^\star=0$, the inter-group gap $\|\bar x_I^\star-\bar x_R^\star\|$ is positive whenever $\mathcal R^\star\neq\varnothing$. Since selection depends only on immutable coordinates, the persistent separation is supported entirely on the immutable subspace: $D^\star=\|\Pi_{J_N}(\bar x_I^\star-\bar x_R^\star)\|$.
The actionable component of the tail centroid vanishes, but an immutable gap remains. Stratification, therefore, corresponds to a fixed partition with a strictly positive immutable distance between groups.

%% file: experiment.tex
\section{Numerical Case Study}
\begin{figure*}\vspace{+2mm}
    \centering
    \includegraphics[width=0.94\textwidth]{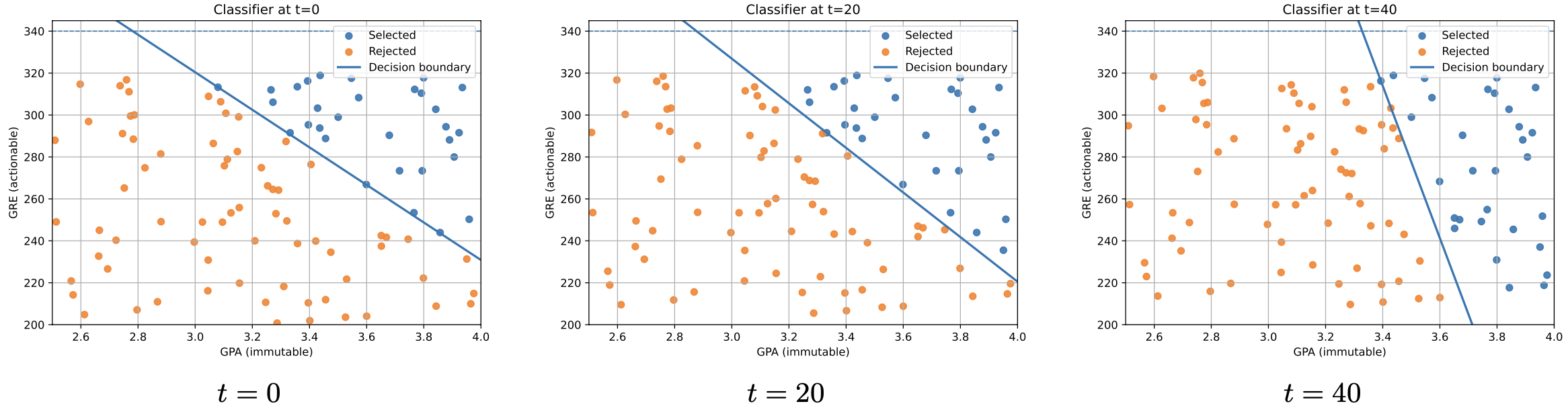}
    \caption{Evolution of the classifier and population over time.}
\label{fig:classifier_evolution}
\end{figure*}

\begin{figure*}\vspace{-3mm}
    \centering
    \includegraphics[width=0.95\textwidth]{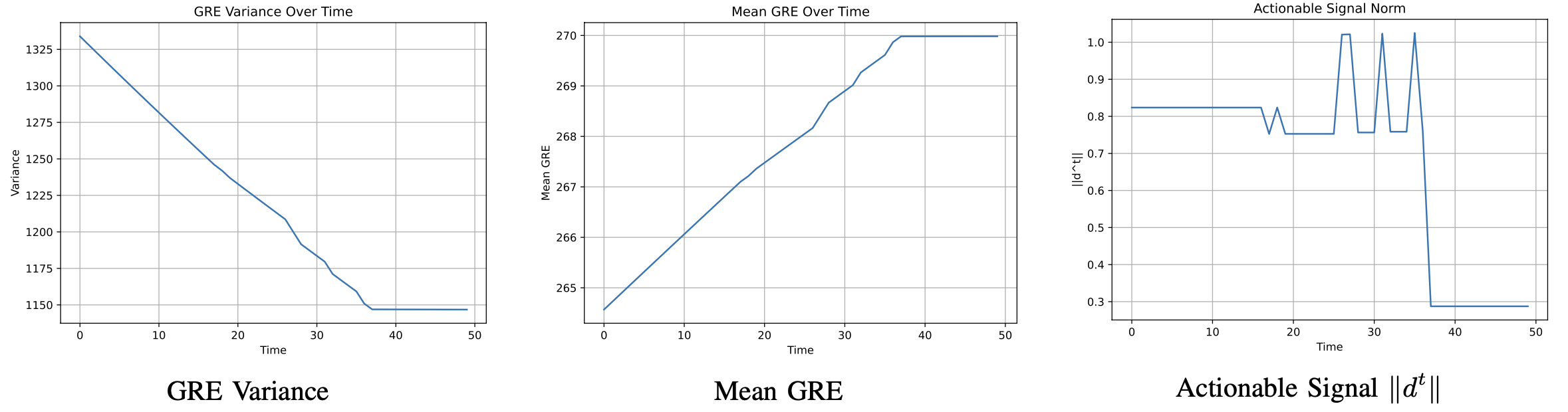}
    \caption{Aggregate evolution of GRE variance, mean GRE, and actionable signal norm.}
\label{fig:aggregate_dynamics}\vspace{-3mm}
\end{figure*}

We illustrate the closed-loop evolution of the CVaR-based classifier in \ref{sec:designer} when candidates update their actionable feature (GRE) according to the logarithmic barrier effort rule derived earlier in \ref{sec:log_barrier}. The actionable feature is bounded above by $\overline{g} = 340$, and the marginal cost of improvement increases as this upper bound approaches.

\subsection{Case Setup}
At each time step, the designer recomputes the scoring vector $w^t$ using the empirical top-$\rho$ tail of the population, which induces an updated decision boundary. Candidates who receive unfavorable decisions respond optimally by adjusting their actionable feature along the direction implied by the current classifier, subject to their effort constraints.
The resulting dynamics aim to capture the interaction among the geometry of the selection rule, candidate’ strategic recourse behavior, and the increasing marginal effort cost generated by the upper bound on the actionable feature.

\subsection{Evolution of the Decision Boundary}

Figure \ref{fig:classifier_evolution} shows the evolution of the decision boundary. At $t=0$, the separating hyperplane is oblique, indicating the joint contribution of GPA (immutable) and GRE (actionable) to the CVaR centroid. In this initial case, rejected candidates below the boundary have a clear incentive to increase their GRE scores. Consequently, from $t=0$ to $t=20$, the rejected candidates move upward along the GRE dimension. 

However, as these candidates improve, the top-$\rho$ tail changes, which in turn causes the classifier to rotate.
This demonstrates the endogenous co-evolution of the population's feature distribution and the decision rule. From $t=20$ to $t=40$, the improvement in GRE scores becomes slow. This occurs primarily because candidates approach the ceiling $\overline g$, and the logarithmic barrier increases the marginal cost of effort. 
By $t=40$, the decision boundary has become steeper, reflecting the diminishing incremental returns of further GRE improvements. 
In general, the system transitions from a phase of rapid adjustment to one of gradual stabilization. 

\subsection{Aggregated Dynamics}

Figure~\ref{fig:aggregate_dynamics} summarizes the aggregated evolution of the system. 
The GRE variance panel shows that dispersion initially persists, showing heterogeneous responses across candidates to the actionable signal. As candidates approach the ceiling, the variance decreases and stabilizes.
The mean GRE panel displays steady upward movement during the early periods, followed by gradual flattening. This flattening is consistent with the endogenous slowdown implied by the first-order condition in Section \ref{sec:log_barrier}. As $g_t \to \overline g$, the barrier term becomes dominant and the optimal effort level $\alpha$ converges toward zero.
The actionable signal panel tracks the norm $\|d^t\|$.  Initially, the signal may remain substantial, thus driving the improvement in GRE scores. Over time, as the population's feature distribution shifts, the geometry of the top-$\rho$ centroid adjusts, causing $\|d^t\|$ to fluctuate before eventually declining.
Convergence may therefore arise from two complementary forces: attenuation of the signal itself or suppression of effort induced by the boundary constraint.

%% file: conclusion.tex
\section{Conclusion and Future Work}

This paper studies actionable recourse in competitive decision environments. While existing work typically analyzes recourse for a single individual responding to a fixed decision rule, many real-world systems select only a limited fraction of applicants. We develop a dynamic framework in which individuals strategically modify actionable features in response to a risk-based selection rule. Since many agents adjust simultaneously, the benchmark for success evolves with the population feature distribution, leading to an endogenous selection mechanism. The proposed model provides an interpretation of involution: as individuals exert increasing effort to improve their profiles, the collective benchmark rises, making further improvement less effective.
Our analysis also shows that individuals who initially receive favorable outcomes determine both the standard of success and the direction in which others must compete. This can amplify initial disparities and generate persistent performance gaps across the population. Future directions include further analysis of the dynamical system and the connection to dynamic games in feedback form.